\documentclass[a4paper,fleqn]{cas-dc}
\usepackage{cas-dfrws} 
\renewcommand{\dfrwsconference}{Digital Forensics Research Conference Europe
(DFRWS EU),  March 29--April 1, 2022}

\usepackage{stix} 

\usepackage{graphicx}
\usepackage{amsmath}
\usepackage{amsthm}
\newtheorem{proposition}{Proposition}

\usepackage[dvipsnames]{xcolor}
\usepackage{mathtools}
\usepackage[colorlinks]{hyperref}
\hypersetup{
    linkcolor=CornflowerBlue,
    citecolor=CornflowerBlue,
    urlcolor=cyan
}

\usepackage[authoryear,longnamesfirst]{natbib}
\newenvironment{work}{[[}{]]}
\newcommand{\hb}{\rightarrow}
\newcommand{\rt}{\textit{rt}}
\newcommand{\event}{\textit{event}}

\usepackage{version}
\excludeversion{fullversion}
\newtheorem{definition}{Definition}

\begin{document}


\title[mode=title]{Defining Atomicity (and Integrity) for Snapshots of Storage in Forensic Computing}
\shorttitle{Defining Atomicity (and Integrity) for Storage Snapshots}

\shortauthors{Ottmann et~al.}


\author[1]{Jenny Ottmann}[orcid=0000-0003-1090-0566]
\cormark[1]
\ead{jenny.ottmann@fau.de}
\credit{Conceptualization, Methodology, Investigation, Writing - Original Draft, Writing - Review and Editing}

\author[2]{Frank Breitinger}[orcid=0000-0001-5261-4600]
\ead{frank.breitinger@unil.ch}
\ead[url]{https://FBreitinger.de}
\credit{Conceptualization, Writing - Review and Editing, Supervision}

\author[1]{Felix Freiling}[orcid=0000-0002-8279-8401]
\cormark[1]
\ead{felix.freiling@fau.de}
\credit{Conceptualization, Methodology, Investigation, Writing - Original Draft, Writing - Review and Editing, Supervision}

\address[1]{Department of Computer Science,
  Friedrich-Alexander-Universit\"at Erlangen-N\"urnberg (FAU),
  Erlangen, Germany}

\address[2]{School of Criminal Justice, 
  University of Lausanne, 1015 Lausanne, Switzerland}

\cortext[1]{Corresponding authors.}

\nonumnote{Copyright remains with the authors.}

\begin{abstract}
  The acquisition of data from main memory or from hard disk storage
  is usually one of the first steps in a forensic investigation. We
  revisit the discussion on quality criteria for ``forensically
  sound'' acquisition of such storage and propose a new way to capture
  the intent to acquire an \emph{instantaneous} snapshot from a single
  target system.  The idea of our definition is to allow a certain
  flexibility into when individual portions of memory are acquired,
  but at the same time require being consistent with causality (i.e.,
  cause/effect relations). Our concept is much stronger than the
  original notion of atomicity defined by \citet{correctness} but
  still attainable using copy-on-write mechanisms. As a minor result, we also
  fix a conceptual problem within the original definition of integrity.
\end{abstract}

\begin{keywords}
storage acquisition \sep instantaneous snapshot \sep correctness \sep integrity
\end{keywords}

\maketitle

\section{Introduction}


Data from storage devices or main memory are crucial pieces of evidence
today. The acquisition of such data usually means to
copy the data from storage to another storage (controlled by the
analyst) in a way that preserves as much of its evidential value as
possible. A common way to define a ``good'' copy is to formulate a set
of quality metrics that capture the intention of forensic soundness.

Attaining good quality copies appears seemingly simple if storage can
be ``frozen''. As an example, there has been little debate about
the classical way to produce a forensic copy of a hard disk using
\texttt{dd} as described by \citet{Carrier2005}. This is in contrast
to the acquisition of main memory which --- apart from approaches that
literally ``freeze'' RAM
\citep{DBLP:journals/cacm/HaldermanSHCPCFAF09} ---- has received
considerable attention if the acquisition target is not a virtual
machine
\citep{correctness,introducing,Inoue:2011:VIT,Campbell:2013:VMA,Lempereur:2012:PAP,evaluatingat}.
In the form of solid state drives, hard disks have turned into
increasingly active devices which has made forensic data acquisition
in the classical sense impossible \citep{Nisbet:2013:FAA}. In practice, circumstances may
prohibit freezing altogether even though it may be technically
feasible. Examples are production servers that cannot be paused for
operational reasons. On such systems, acquisition is often improvised
and part of a live analysis.

\subsection{Inconsistencies in RAM acquisition} \label{sec:introincons}

As mentioned, quality criteria for data acquisition have most
often been discussed in the context of volatile memory because of the
common problems that occur if RAM is acquired inconsistently.

\begin{fullversion}
  Main memory analysis promises the retrieval of valuable data even on
  systems which encrypt persistent data. However, the quality of the
  acquired main memory snapshot significantly influences how much data can
  be retrieved. If it is impossible to freeze the system when taking
  the memory snapshot, a number of problems can occur because of
  concurrent changes to the memory.
\end{fullversion}
Page smearing, for example, is a common problem on systems under heavy
load or with more than 8 GB of RAM according to \citet{pathforw}. When
page smearing occurs, the page tables in the memory snapshot are not
consistent with the contents of the physical pages because changes
were made to the referenced physical pages after the page tables were
acquired. This can result, for example, in pages being attributed to
wrong processes or inconsistencies in kernel data structures.  Some of
these inconsistencies can lead to problems during the memory analysis
or hinder an analysis completely if it uses kernel data structures. 
But not all inconsistencies have to be apparent during an
analysis. Therefore, precise measurement criteria, to define the
quality of a memory imaging method, can help to evaluate tools
without having to rely on visible problems.

\begin{fullversion}
\citet{introducing} define three different types of inconsistencies which
can occur in time inconsistent snapshots: \emph{fragment
inconsistency}, \emph{pointer inconsistency}, and \emph{value inconsistency}.
The first type of inconsistency concerns objects which are spread over multiple
pages. If such objects are acquired time inconsistently, the contents of the
different pages do not fit together, some pages could contain older values while
others already contain updated values. Pointer inconsistencies concern the page
in which the address of a reference is stored and the page that stores the
referenced contents. If they are acquired time inconsistently, the reference
might, for example, point to a page that has not been allocated yet. The last type
of inconsistency concerns the contents of a page that describe those of another
page. In a time inconsistent snapshot the description might not fit the contents
of the described page anymore. For example, the number of entries in an array
stored in another object is inconsistent with the actual number of entries in
the array in the snapshot.
In their evaluation they observed the influence of time on the quality of the
snapshots.
As an example for inconsistencies in kernel data structures, they
compared the number of stored virtual memory areas
(VMAs) of a process in two different kernel data structures to the VMA counter of the
process. Since a mismatch does not neccessarily imply that the
analysis will be impeded, they also took a closer look at the acquired VMAs to check
if all expected segments, in this case \emph{code} and \emph{stack}, could be
found. They report that in more than 80\% of the analyzed snapshots inconsistencies
were found and that in some cases code and/or stack segment could not be found.
For a closer look at the possible implications of inconsistencies on
the analysis of process address spaces they tried to reconstuct the stack of a
process. The contents of the acquired stack pages were compared to an
atomically acquired view on the stack that was created each time one of the
stack pages was imaged. In their experiment an uncorrupted backtrace could only
be created in one out of ten attempts.
They conclude that information about the time at which the pages were acquired
should be made available to analysts because it plays an important role for the
quality of a memory snapshot. They also implemented changes to show include the
gathering of the information in LiME and it display in Volatility.
Additionally, they performed experiments with a changed implementation of
LiME in which not all pages are acquired sequentially. Instead, the acquisition starts
with pages containing important information for forensic analysis before proceding sequentially for the
remaining address space. Using the adjusted implementation they repeated their
experiments and were able to gather all code and stack segments for the tested
processes as well as correct backtraces.
\end{fullversion}

\subsection{The quest for suitable quality criteria}

To be able to qualify the effects of RAM acquisition, 
\citet{correctness} introduced three criteria to evaluate the quality of a
memory snapshot: \emph{correctness}, \emph{integrity}, and \emph{atomicity}. A
correct snapshot contains all memory regions of the main
memory and for each region exactly the value it had in main memory at the
time of the acquisition. Thus, to achieve correctness not only a correct
implementation is necessary but the used system components must return the
correct values as well. The integrity criterion focuses on memory content
changes between the start of the acquisition process and the acquisition of each
memory region. Integrity is violated for memory content that changed  after
the acquisition was started and before it could be copied by
the acquisition process. The atomicity criterion in
contrast allows changes of memory contents if the acquired memory
regions are causally consistent. This means that no memory region content in the
snapshot was influenced by changes to a memory region that are not part of the
snapshot. 

Recently, \citet{introducing} criticized the atomicity definition by
\citet{correctness} for being ``extremely difficult to measure in
practice''. Instead, they suggested a criterion called \emph{time
  consistency}. A snapshot is time consistent if there ``exists a
hypothetical atomic acquisition process that could have returned the
same result''.  However, they do not provide a precise formalization
of this concept.

\subsection{Related work} 

There exists a large body of work that investigates the creation of
snapshots in distributed concurrent systems 
often with the aim to detect predicates on the state of a distributed
computation \citep{DBLP:journals/dc/ChaseG98}.  In this work, focus
has  been on \emph{asynchronous} distributed systems where the
best available notion of time is causality
\citep{virtual,DBLP:journals/dc/SchwarzM94}. In such systems,
concurrent execution of events makes the global state
``relativistic'', i.e., it is often not possible to exactly say in
which state the system is or has been
\citep{DBLP:conf/pdd/CooperM91,DBLP:conf/srds/GartnerK00,DBLP:conf/ACISicis/ChuB08}. 


In forensic computing, data acquisition is (currently) usually
performed in a synchronous environment. While concurrency arises even
in such systems from different threads that operate on shared memory,
such systems provide a common centralized clock that can be used
to order events.
%
%
If events can potentially be totally ordered, the sequence of global
states through which the system progresses is well
defined. In contrast to the assumptions made in previous theoretical work that describes
algorithms for predicate detection in synchronizable systems
\citep{DBLP:journals/dc/Stoller00}, real systems usually do not keep
track of timestamps of individual events.
%
%
The application of complex snapshot algorithms in forensics appears
not advisable anyway since taking forensic snapshots should minimize
interference with the observed system.
So, while the literature on distributed systems gives many insights
into the problem area, we are not aware of work that is of direct
help.


Other related work is concerned with measurement of the quality of
snapshots. Early work avoided the need to define quality criteria by
simply comparing the output of a tool with the memory content of the
machine from which the snapshot was taken
\citep{Inoue:2011:VIT,Campbell:2013:VMA,Lempereur:2012:PAP}.
%
%
\citet{evalplat} were the first to perform a practical
evaluation of memory acquisition tools against the abstract quality
criteria of \citet{correctness}. They implemented the evaluation
platform using \emph{Bochs} and took a white-box testing
approach. 
\begin{fullversion}
  First the code of the tested tools was examined and hypercalls were
  inserted to enable the host process to monitor the acquisition
  process. It is signaled to the evaluation platform among others when
  the acquisition process is started and ended, when the imaging of a
  page begins and is completed, and when the acquisition program is
  loaded into memory and unloaded.
\end{fullversion}
With the help of \begin{fullversion}the\end{fullversion} inserted
hypercalls, three tools were evaluated. Correctness could be measured
exactly by comparing the memory image created by the acquisition
process to an image created in parallel by the host. Atomicity could
not be measured exactly as this would have required to keep track of
all causal dependencies in the guest, a task deemed nearly infeasible
by the authors. Instead, possible atomicity violations were measured by
keeping track of which threads accessed already acquired pages and
then modified a page that was not already acquired.  Therefore, the
results present an upper bound of atomicity violations. Integrity was
estimated by comparing a memory image taken by the host shortly before
the acquisition process was loaded into the guest memory with one
taken by the host shortly after the acquisition process finished. 
\begin{fullversion}
  The results showed that the tools acquired the memory correctly. In
  the case of atomicity the percentage of pages for which atomicity
  was possibly violated was higher for larger memory sizes and reached
  up to 75\%. They attributed this mainly to the increased time
  necessary to complete the acquisition process. The amount of
  integrity violations on the other hand was smaller the larger the
  assigned guest memory was. This can be explained by the fact that
  the system load was not increased as well.  Apart from the
  limitations regarding the methods to evaluate atomicity and
  integrity, the platform is also limited by the relatively small
  amount of memory that can be assigned to a guest, 2 GB max, as well
  as the inability to evaluate 64-bit or closed source applications.
\end{fullversion}

Building on the results by \citet{evalplat}, \citet{evaluatingat} took
correctness for granted and followed a black-box approach to measure
atomicity and integrity.
Because their method does not rely on modifying the source code
of tools, more tools could be evaluated, including direct memory access (DMA) and cold boot.
\begin{fullversion}
  In total 12 tools, relying on different snapshoting methods,
  e.g. kernel-level acquisition and DMA, were tested.
\end{fullversion}
For the tests they wrote a program which allocates sequentially
numbered memory regions and one to extract the numbered regions from a
memory snapshot. The numbers serve as a counter that allows to
\emph{estimate} the level of atomicity and integrity. 
\begin{fullversion}
  This is done by calculating the atomicity and integrity delta
  respectively. The time elapsed between the acquisition of the first
  memory region and the last memory region is the atomicity delta
  because the longer an acquisition takes, the more atomicity
  violations become likely. Since integrity is defined dependent on
  the time at which the acquisition was started, the integrity delta
  is the average time elapsed between the start of the acquisition and
  the imaging of each memory region.  The results of their experiment
  show that, as expected by the authors, methods that allow to freeze
  the state of the system achieve the highest atomicity and
  integrity. For the other methods, tools belonging to the same
  category achieved similar results.
\end{fullversion}

\subsection{Contributions}

In this paper, we revisit \citet{correctness} and follow the demand
formulated by \cite{introducing} for more ``permissive'' quality
metrics for the acquisition of storage: We formalize two new
definitions of atomicity which we call \emph{instantaneous} and
\emph{quasi-instantaneous consistency}. Both can be seen as possible
formalizations of the notion of ``time consistent'' by
\citet{introducing}.

Instantaneous consistency resembles the quality of an ``ideal''
snapshot taken from a frozen system and implies quasi-instantaneous
consistency. But although being slightly weaker in guarantees, a
quasi-instantaneous snapshot is indistinguishable from an
instantaneous snapshot. We show that quasi-instantaneous
snapshots can be achieved (by performing memory snapshots using the
idea of copy-on-write). Moreover, under certain assumptions quasi-instantaneous consistency
implies the (classic) causal consistency of \citet{correctness} so
quasi-instantaneous snapshots do not violate cause-effect
relations. Since the common memory snapshot techniques based on
software generally do not even guarantee causal consistency, we also
raise the question of how to assess a memory snapshot regarding its
level of atomicity.

As a minor contribution, we propose a new definition of integrity that
is refined from \citet{correctness} and removes some of its
theoretical weaknesses. We formulate all concepts
independent from concrete storage technologies so that they can be
applied to any block-based digital storage, be it volatile or
persistent.

\subsection{Outline}

The following section introduces the model used to formalize our
concepts.  Section~\ref{sec:qualcrit} continues with a formal
definition of the two new forms of atomicity.  Section~\ref{sec:achieving} 
provides an overview of methods with which these
consistency criteria can be achieved, while
Section~\ref{sec:measuring} discusses some ideas on how to evaluate snapshots of storage with 
respect to the new metrics. Section~\ref{sec:integrity}
briefly presents our result on the notion of
integrity. Old and new concepts are compared in Section~\ref{sec:relations}.
Section~\ref{sec:legal}  discusses legal implications
of our concepts for concrete investigations. 
Finally, Section~\ref{sec:concl} concludes the paper.

\section{Model}
\label{sec:model}

We now define the notation to describe computations on memory regions
and the timing relations of the events that happen therein (the
definitions are adapted from \citet{DBLP:conf/icdcs/ZhengG19} using
the timing notation of \citet{DBLP:conf/ACISicis/ChuB08}).

\subsection{Processes, memory regions and events}

We consider a finite set $P=\{p_1,\ldots\}$ of processes (or threads)
that perform operations on a set of $n$ memory regions
$R=\{r_1,\ldots,r_n\}$.
%
%
Performing an operation results in an event $e=(p,r)$, where $e.p$
denotes the process that performed $e$ and $e.r$ denotes the memory
region on which the process performed the operation.
The set of all events is denoted $E$.
We assume that operations on a single memory region are performed
sequentially (e.g., by using hardware arbitration or locks).

\subsection{Space/time diagrams and cuts}

\begin{figure}
  \centering
  \includegraphics[scale=0.95]{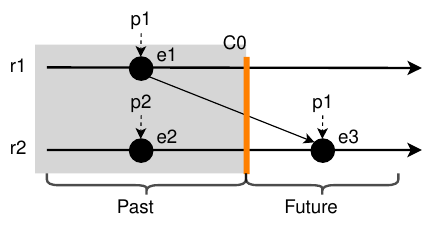}
  \caption{Space/time diagram of a computation and one possible cut
    C0. The events to the left of the cut are part of the past, those
    to the right part of the future.}
  \label{fig:simplecut}
\end{figure}

We use space/time diagrams, commonly used in distributed computing
\citep{virtual}, to depict computations. The sequence of
events within the memory regions serves as timeline from left to right,
and the sequential activities of processes are depicted as arrows that
connect events. An example is shown in Fig.~\ref{fig:simplecut} with
two memory regions $r_1$ and $r_2$, a process $p_1$ executing events
$e_1$ and $e_3$ and a process $p_2$ executing event $e_2$.

A \emph{cut} through the space/time diagram is indicated by a line
that intersects each memory region exactly once. Formally, a cut is a
subset of events of the computation and can be regarded as separating
the events into a ``past'' (to the left of the cut) and a ``future''
(to the right of the cut). Fig.~\ref{fig:simplecut} shows an example of a cut
through a computation.

\subsection{Causal order on events}

A computation is modeled as a tuple $(E,\hb)$ where $E$ is the set of
events and $\hb$ is the \emph{causal order} on $E$, i.e., the smallest
transitive relation such that
\begin{enumerate}
\item if $e.p=f.p$ and $e$ immediately precedes $f$ in the sequential
  order of that process, then $e\hb f$, and
\item if $e.r=f.r$ and $e$ immediately precedes $f$ in that memory
  region, then $e\hb f$.
\end{enumerate}
The order $\hb$ corresponds to Lamport's happened-before relation
\citep{DBLP:journals/cacm/Lamport78}.

The order $\hb$ merely encodes which events might have influenced
which other events, i.e., if $e\hb f$ or $f\hb e$ then either $e$ may
have influenced $f$ or vice versa. However, if neither $e\hb f$ nor
$f\hb e$ we say that $e$ and $f$ are \emph{concurrent}, i.e., it is
not possible to order the two events regarding causality.

\subsection{Observability of causal relations}

The causal order relation between events is by definition independent
of the concrete values that processes read from or write to memory
regions. For example, in the computation depicted in
Fig.~\ref{fig:simplecut} all events $e_1$, $e_2$ and $e_3$ could be
read events that do not modify the content of the memory regions. 
Causal dependencies, therefore, may not be observable
unless they are somehow reflected by the stored values. A minimum
requirement for events to be observable is that they perform a state
change of the memory region, e.g., to change the stored value from 0
to 1. Events that always update the state of a memory region to a
different state as before are called \emph{modifying events}.

Merely requiring that events modify the state of a memory region does
not imply that state changes can always be observed. The reason for
this is that subsequent state changes can annihilate effects of
previous state changes. For example, event $e_2$ in
Fig.~\ref{fig:simplecut} could change the value of memory region $r_2$
from 0 to 1, and event $e_3$ could change it back from 1 to 0. The
fact that an event has occurred is not observable if the starting and
ending state of $r_2$ is inspected. This can be avoided by demanding
that every event assigns a ``fresh'' value to the memory region. This
is the case, for example, if the stored value is a counter that is
incremented with every event. Events that change the value of the
memory region to a new unique value are called \emph{uniquely
  modifying events}.

Techniques to observe causal relationships in distributed systems
(like vector clocks \citep{virtual}) are commonly based on the
assumption of uniquely modifying events. We will revisit these
concepts later when exploring the compatibility between our new
consistency notions and causal order.

\subsection{Consistent global states}

Cuts are often considered as representations of global states of the
computation.
Fig.~\ref{fig:lattice3} depicts 
multiple possible cuts through the computation shown in
Fig.~\ref{fig:simplecut}. For example, cut $c_0$ is the
initial cut (no event has happened yet), $c_1$ is the cut where $e_1$
is the only event that has happened, and $c_4$ is the cut where $e_2$
and $e_3$ have happened but not $e_1$. 
The causal order $\hb$ on events induces a partial order
on these states that form a lattice. The lattice of all such global
states of the computation shown in Fig.~\ref{fig:lattice3} is depicted in
Fig.~\ref{fig:lattice4}.

\begin{figure}
  \centering
  \includegraphics[scale=0.95]{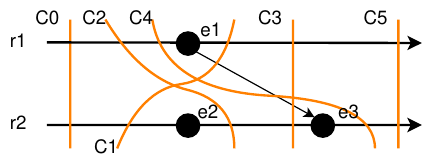}
  \caption{Space/time diagram and possible cuts of a computation.}
  \label{fig:lattice3}
\end{figure}

\begin{figure}
  \centering
	\includegraphics[scale=0.95]{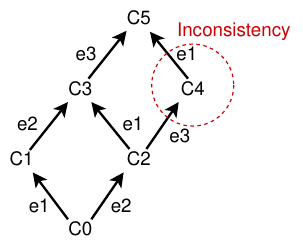}
  \caption{Lattice of global states of the computation depicted in
    Fig.~\ref{fig:lattice3}.}
  \label{fig:lattice4}
\end{figure}

Note, in the absence of any notion of real-time, it cannot be
determined whether event $e_1$ happened before $e_2$ or not (with
respect to $\hb$ they are not ordered). Hence, it cannot be
determined which sequence of global states occurred in the
computation, as long as the cut respects the causality relation
$\hb$. Cut $c_4$ is one that violates $\hb$ in that events $e_2$ and
$e_3$ are contained in the global state while $e_1$ is not. This 
cannot have happened since $e_3$ is the effect of $e_1$, i.e., if
$e_3$ is contained in the cut, then $e_1$ also must be. This is the
basis of the definition of a \emph{consistent cut}.

\begin{fullversion}
  Formally, a consistent cut $c$ is closed with respect to $\hb$, i.e.,
  $$\forall e\in c: e'\hb e \Rightarrow e'\in c.$$
\end{fullversion}

\subsection{Realtime}

In systems where real-time clocks are available, it may be possible to
order two events $e_1$ and $e_2$ in Fig.~\ref{fig:lattice3} by
comparing the real-time readings of when they occurred. For an event
$e$ we denote by $\rt(e)$ the real-time clock reading when $e$
occurred. Formally, $\rt$ is a function mapping the set of events to
the time domain $T$. For simplicity and without loss of generality, we
equate $T$ with the set of natural numbers.
Using $\rt$, it is possible to transform the partial order $\hb$ into a
total order by ordering every event $e\in E$ using $\rt(e)$.

\subsection{Snapshots}

The effects of events are possible value changes in the memory
regions. Following the notation of \citet{correctness}, we define the
set of all possible values of a memory region as $V$.
The contents of the memory regions can then be formalized as a
function $m: R \times T \rightarrow V$ that returns the value of a
specific memory region at a specific point in time. Function $m$
encodes a form of \emph{ground truth} of what values the
memory contained at any specific time.

Informally, a \emph{snapshot} is a copy of all memory regions. Since
individual memory regions might be copied at different points in time,
we formalize a snapshot as a function $s: R \rightarrow V \times T$,
i.e., for every memory region we store the value and the time this
value was copied from memory. We denote by $s(r).v$ the value stored
for region $r$ in snapshot $s$ and by $s(r).t$ the corresponding
time. For example, if $s(r_1)=(15, 3)$ then memory region $r_1$ was
copied at time $s(r_i).t=3$ with a value of $s(r_i).v = 15$.
Snapshots correspond to cuts through the space-time diagram of a
computation.


Taking a snapshot of a computation means to copy the current values
from memory regions into the snapshot, but a snapshot does not contain
any references to events that have happened. To be able to formally
connect events in and snapshots of a computation, we introduce one
additional notation: For a real-time value $t$ and a memory region $r$
we denote by $\event(r,t)$ the \emph{most recent event} that happened
on memory region $r$ at a time before or equal to $t$. Formally,
$\event$ is a function $\event : R\times T \rightarrow E$. If
$\event(r,t)=e$ then $\rt(e)\leq t$ and there exists no other event on
$r$ that happened between $\rt(e)$ and $t$.

Technically, the set of events contained in the cut corresponding to a
snapshot $s$ consists of all events that lie to the left of
$\event(r, s(r).t)$ (including the event itself).

\section{Defining Atomicity}
\label{sec:qualcrit}

Intuitively, atomicity is a notion to characterize the degree of
freedom of signs of concurrent activity. High atomicity therefore
attempts to bound the effects that arise from an observation taking
place concurrently to a computation.

\subsection{Causal consistency}

The original definition of atomicity introduced by \citet{correctness}
is based on the causal dependency relation $\hb$ between events. It
states that the set of events derived from a snapshot corresponds
to a consistent cut. The rationale behind this definition was that
snapshots should respect causality, i.e., for each effect the snapshot
contains its cause. The definition rules out any inconsistent cuts as
allowed snapshots. Such an example is depicted in
Fig.~\ref{fig:causincons} where two events $e_1$ and $e_2$ occurred on
region $r_1$ and $r_2$ respectively and $e_1\hb e_2$. A snapshot that
contains the contents of $r_1$ before $e_1$ happened and the contents
of $r_2$ after $e_2$ happened is causally inconsistent because the
change introduced by $e_1$ that caused $e_2$ is missing.

\begin{figure}
  \centering
  \includegraphics[scale=0.95]{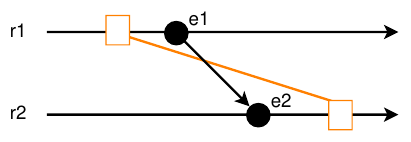}
  \caption{Causally inconsistent snapshot.}
  \label{fig:causincons}
\end{figure}

\citet{correctness} argued that snapshots should at least be
consistent with causality because causally inconsistent snapshots
clearly cannot have happened. Unfortunately, many software-based
snapshot approaches for RAM do not produce even causally consistent
snapshots.

\subsection{Instantaneous consistency}

Causally consistent snapshots guarantee that snapshots respect causal
relationships. However, causal consistency is a notion defined for
asynchronous distributed system, i.e., systems where no notion of
real-time exists and time is reduced to causality. In such systems,
events can be re-ordered along the sequential timeline if
causal relationships are respected. Therefore, \emph{every} consistent
global state is a state that the computation \emph{potentially} could
have passed through.
In practice, and in particular in those systems that we focus on here
(smartphones, PCs, servers), often a notion of real-time exists that
allows to narrow down the set of consistent global states that
\emph{actually} have happened \citep{DBLP:journals/dc/Stoller00}.

Based on these insights, we now define an idealistic (and much
stricter) consistency criterion based on the time at which the memory
regions are copied.  The notion of \emph{instantaneous consistency}
formalizes the idealistic intent of snapshots in which all memory
regions are copied at exactly the same time.

\begin{definition}[instantaneous consistency]
  A snapshot $s$ satisfies \emph{instantaneous consistency} iff all
  memory regions in $s$ were acquired at the same point in
  time. Formally:
  $$ \forall r, r'\in R: s(r).t = s(r').t$$
\end{definition}

If a snapshot satisfies instantaneous consistency, we say that the
snapshot \emph{is} instantaneous.  Obtaining instantaneous snapshots
is possible if the system of which the memory contents are extracted
can be paused, for example when the main memory of a virtual machine
is dumped. An example of an instantaneous snapshot is depicted in
Fig.~\ref{fig:insta}.

\begin{figure}
  \centering
  \includegraphics[scale=0.95]{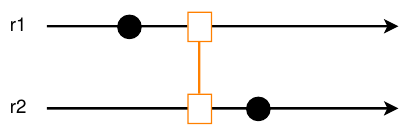}
  \caption{When an instantaneous snapshot is taken, all memory regions
    are copied at the same time.}
  \label{fig:insta}
\end{figure}

From a forensic point of view, it is desirable that a snapshot is
instantaneous because it resembles something that is easy to
understand and ``obviously'' free of any problems of
concurrency. This aspect is important for legal proceedings in which
doubts on the way evidence was gathered can severely degrade
its evidential value.

\subsection{Quasi-instantaneous consistency}

Taking instantaneous snapshots usually requires freezing the system
from which memory is copied, at least this is the case for systems
where no inherent (hardware) mechanism exists to copy all memory
regions at the same time. So taking instantaneous snapshots in
practice is hard, if not impossible. 

We therefore define a slightly weaker criterion that captures the
main ideas of instantaneous consistency while allowing to take
snapshots without freezing the system. We call this
\emph{quasi-instantaneous consistency}.


\begin{definition}[quasi-instantaneous consistency]	
  A snapshot $s$ satisfies \emph{quasi-instantaneous consistency} iff
  the values in the snapshot could have also been acquired with an
  instantaneous snapshot $s'$. Formally:
  \begin{multline*}
    \exists s': (\forall r, r'\in R: s'(r).t = s'(r').t) \mathrel{\land} \\
    \forall r \in R: s'(r).v = s(r).v
  \end{multline*}
\end{definition}

The above definition does not require that the snapshot 
\emph{is} taken instantaneously but that it \emph{could} have
been taken instantaneously, i.e., that the values of all memory
regions in the snapshot were coexistent in memory at (at
least) one point in time during the acquisition.

Fig.~\ref{fig:quasi} shows an example of a snapshot which is
quasi-instantaneous. In this example a point in time can be found at
which the contents of the two memory regions in the snapshot were
coexistent in memory. 
When such a point in time cannot be found, the snapshot is not
quasi-instantaneous. Assuming modifying events, Fig.~\ref{fig:notok}
shows an example for this case. Because of the order of the events
$e_1$ and $e_2$ and the time points at which the memory regions were
added to the snapshot, the snapshot contains values that were never
coexistent in main memory. In this case it is impossible to find a
time at which a snapshot containing the same values could have been
taken instantaneously.

\begin{figure}
  \centering
  \includegraphics[scale=0.95]{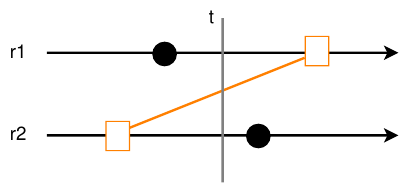}
  \caption{When a snapshot satisfies quasi-instantaneous consistency,
    a point in time can be found at which the contents of the memory
    regions in the snapshot coexisted in the copied main memory. The
    same result would have been achieved with an instantaneous
    snapshot at time $t$.}
  \label{fig:quasi}
\end{figure}

\begin{figure}
  \centering
  \includegraphics[scale=0.95]{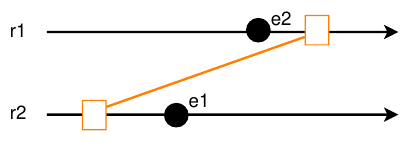}
  \caption{A snapshot that is not quasi-instantaneously consistent if $e_1$ and $e_2$ modify the values of $r_1$ and $r_2$, respectively.}
  \label{fig:notok}
\end{figure}

\begin{fullversion}
  Alternative definition of quasi-instantaneous consistency in analogy
  to the closure definition of causal consistency:

  $S$ is the set of events belonging to snapshot $s$ (see above). Then
  $s$ is quasi-instantaneous iff
  $$\forall e\in S: \rt(e') \leq \rt(e) \Rightarrow e'\in S$$ 
  
  Another way to see this is that the ``latest'' event (w.r.t.~$\rt$)
  defines the time of the snapshot.
\end{fullversion}

\subsection{Relations between the consistency definitions}



Instantaneous consistency is the strongest concept of the presented
consistency definitions. If all memory regions can be copied at the
same time, no inconsistencies can arise due to concurrent
activity. Therefore, instantaneous consistency implies
quasi-instantaneous consistency.

The relation between quasi-instantaneous and causal consistency is
slightly less apparent. We first argue that a causally consistent
snapshot is not necessarily quasi-instantaneously consistent. To see
this, consider the computation in Fig.~\ref{fig:notok} and note that
the events $e_1$ and $e_2$ are independent of each other. Therefore,
any snapshot of this computation is causally consistent. However, if
$e_1$ and $e_2$ are modifying events, the snapshot is not
quasi-instantaneously consistent. So causal consistency does not
generally imply quasi-instantaneous consistency.

But what about the inverse question, i.e., is every
quasi-instantaneous snapshot also causally consistent? Interestingly,
the answer to this question depends on the nature of events that
determine causal consistency. To see this, consider the computation in
Fig.~\ref{fig:quasiatomic} which is similar to the one depicted in
Fig.~\ref{fig:notok} but where events $e_1$ and $e_2$ have a causal
relationship. If neither $e_1$ nor $e_2$ are modifying events then the
values stored in memory do not change and so any snapshot would be
quasi-instantaneous, also the one depicted in
Fig.~\ref{fig:quasiatomic}. So in this case, a snapshot might be
quasi-instantaneously consistent but still causally inconsistent. But
even if we only have modifying events, the changes of $e_1$ or
$e_2$ could be reverted and the resulting quasi-instantaneously consistent
snapshot might again not be causally consistent. But if we
assume that we only have uniquely modifying events, this cannot happen
anymore.

\begin{proposition}
  If all events are uniquely modifying, then any quasi-instantaneously
  consistent snapshot is also causally consistent.
\end{proposition}

\begin{proof}
  Let $s$ be a snapshot satisfying quasi-instantaneous
  consistency. From the definition of quasi-instantaneous consistency
  follows that there exists an instantaneous snapshot $s'$ that
  contains the same values as $s$ for every memory region. Since $s'$
  is instantaneous, it is also causally consistent. But since all
  events are uniquely modifying, no events can have happened between
  $s'$ and $s$. Therefore, $s$ is also causally consistent.
\end{proof}

\begin{figure}
  \centering
  \includegraphics[scale=0.95]{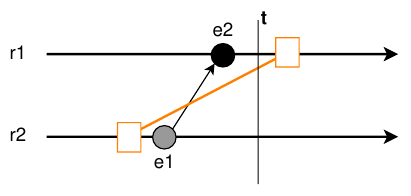}
  \caption{If $e_1$ is an event that does not change the memory
    contents, then the snapshot is quasi-instantaneously consistent
    but not causally consistent. If $e_1$ and $e_2$ are modifying
    events, then the snapshot is neither quasi-instantaneously nor
    causally consistent.}
  \label{fig:quasiatomic}
\end{figure}

As observed by \citet{introducing}, causal consistency is very
permissive but appears to be the smallest common denominator of any
acceptable quality measure of atomicity. However, it is too permissive
to be easily attainable. Instantaneous consistency, the ideal notion
of atomicity, is too strong. Quasi-instantaneous consistency is an
intermediate definition that does not need to halt the system but
still can express a similar level of instantaneousness. It is close in
spirit to \citeauthor{introducing}'s concept of \emph{time
  consistency}, which is satisfied ``if there was a point in time
during the acquisition process in which the content of those pages
co-existed in the memory of the system'' \citep{introducing}.

\section{Achieving Consistency}
\label{sec:achieving}

\begin{fullversion}
  Whenever freezing a system is not a possiblity and the system has to keep
running during the acquisition, achieving atomicity is challenging. As prior
research showed, atomicity violations are very
likely to occur in memory snapshots that are created on a running system. 
\end{fullversion}

One possibility to achieve quasi-instantaneous consistency of
snapshots created on a running system is to ensure that, after the
acquisition process has been started, no memory content will be
modified before it has been copied. This method is known as
\emph{copy-on-write} in the area of systems software. It is rather
easy to see that a snapshot created using this technique satisfies
quasi-instantaneous consistency, because any state changes occurring
after the start of the acquisition are not included into the memory
snapshot. Therefore, the contents in the snapshot are equal to those the
memory regions contained at the start of the acquisition.

An example can be seen in Fig.~\ref{fig:cow}: Because $e_2$ would have
been executed on memory region $r_2$ after the start of the
acquisition process but before it was copied, the event is interrupted
and the region is copied first. Afterwards the event can be executed.

\begin{figure}
  \centering
  \includegraphics[scale=0.95]{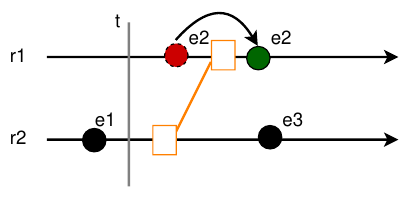}
  \caption{After the acquisition has been started at time t, no changes to
	memory regions that have not been copied yet are allowed. Once a region
	has been copied its content may be changed by events.}
  \label{fig:cow}
\end{figure}

Manipulating the page table entries is a convenient method to do
this. By taking away the write permission of all page tables entries,
trying to change the page will result in a page fault that can be
handled accordingly. But the system of which the memory snapshot is
created should not be manipulated to such a great extent. Therefore,
instead of manipulating the page tables of the operating system, a
hypervisor can be used. By taking away the write permissions of the
guest pages on the hypervisor level, write accesses will cause an exit
to the hypervisor. Then the page can be copied and the write
permission for the page can be turned on again. Over the last years
the technical means to implement the technique changed, as can be seen
in the works of \citet{martignoni2010, yu2012} and
\citet{kiperberg2019}. As it cannot always be expected that a system
is already virtualized, methods for the ``on the fly'' virtualization
of a system have also been proposed \citep{hyperleech}.

\section{Measuring Consistency}
\label{sec:measuring}

While theoretical quality criteria are an important step towards
understanding which factors influence the usefulness of a memory
snapshot, the question remains how these criteria can be
measured. Because of the limitations of previous measurement
approaches, we describe an alternative method to evaluate a snapshot
with regard to causal consistency.

Since it is difficult to trace all causal relationships in a system,
we suggest to only keep track of causal relationships in a part of the
system. Tracking causal relationships within one process is a
manageable task. It also allows to perform the evaluation for closed
source tools and on different operating systems. The idea is to use a
simple test program in which memory regions are allocated, and
causally dependent changes on the regions, observed using vector
clocks. If quasi-instantaneous consistency should be measured instead,
realtime timestamps can be used in place of vector clocks.

\subsection{Using vector clocks}\label{sec:vclocks}

\begin{figure}
  \centering
  \includegraphics[scale=0.8]{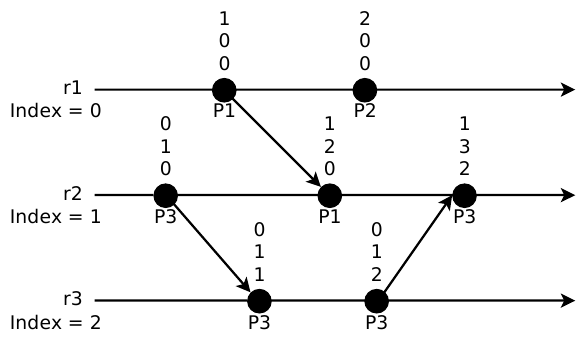}
  \caption{Each memory region has an assigned vector clock and an
    index for the local counter. Three processes, $p_1$, $p_2$, and
    $p_3$, interact with the memory regions. Whenever a process
    accesses a memory region the vector clock the process saw last and
    the local vector clock are combined and the counter at the index
    of the region is increased.}
  \label{fig:vc}
\end{figure}

Vector clocks are a concept from distributed computing that allows to
track logical time by ordering events \citep{virtual}. Usually, vector
clocks are assigned to different processes in a distributed
system. When events are executed by a process or messages between
different processes are exchanged, the clocks need to be updated. As
we want to observe changes on memory regions, all subsequent examples
assign vector clocks to memory regions not processes, the original
definitions by \citet{virtual} are adapted accordingly.

The idea is to assign a counter to each memory region that
increases every time an event (i.e., an access by a process) is
executed on the region. Such counters allow to track causal
relationships between events in the following way: Additionally, to the
local counter, each memory region's vector clock has fields for all
other memory regions' counters.  If we assume a system with
$n$ memory regions, each region has a vector clock (a vector of
counters) $C$ of size $n$. Each region is assigned a unique index to
this vector at which its local counter is found. Whenever a process
accesses a memory region it saves the region's vector clock. When it
accesses the next region the two vector clocks are compared and for
each index the higher value is chosen. Then the local counter is
incremented.

Causal relationships can be detected with vector clocks by ordering them
using the happened-before relation \citep{virtual}:
For two vector clocks, \(C_i\) and \(C_j\), \(C_i < C_j\) holds iff
\begin{multline*}
  \forall x \in \{1, \ldots, n \} : C_i[x] \leq C_j[x]) \\
        \mathrel{\land} (\exists x: C_i[x] < C_j[x])
\end{multline*}
Whenever this does not hold for two vector clocks, the causing
events are concurrent to each other.  Fig.~\ref{fig:vc} shows an
example for three memory regions and their assigned vector
clocks. Each time a process accesses a memory region the vector clocks
are updated. Using the definition we can, for example, see that the
event caused by process $p_2$ on region $r_1$ is only causally
dependent on the event caused by $p_1$ on the same memory region and
concurrent to all other events.

\begin{figure}
  \centering
  \includegraphics[scale=0.8]{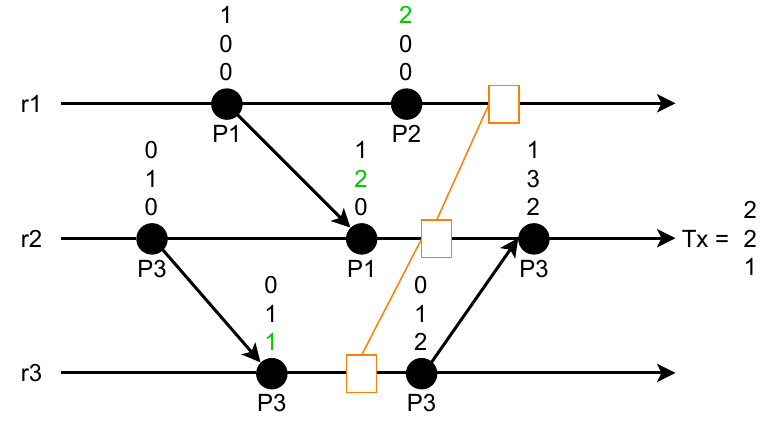}
  \caption{In a consistent snapshot the values of the global time at each
	region's index and that region's vector at the index are equal.}
  \label{fig:vccons}
\end{figure}

\begin{figure}
  \centering
  \includegraphics[scale=0.8]{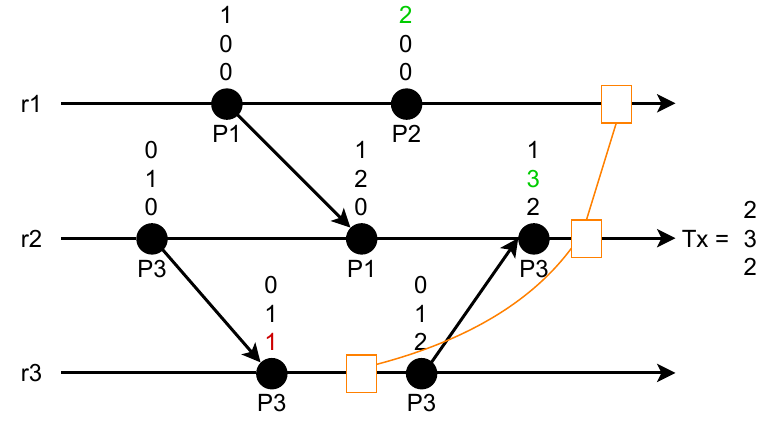}
  \caption{This snapshot is inconsistent because the event on region
    $r_3$ on which the last event included in the snapshot on region
    $r_2$ is causally dependent is not included as well. This becomes
    evident when comparing the vector clock of $r_3$ with the global
    time.}
  \label{fig:vcincons}
\end{figure}

With the help of the vector clocks, inconsistencies in a snapshot (or
cut) can be found.  First, the \emph{global time} vector \(t_s\) of
the snapshot $s$ needs to be calculated. This vector consists of the
highest value for each index in all vector clocks \citep{virtual} as
$$ t_s = sup(C_1,...,C_n). $$
Next, each region's vector clock is compared to the global time. More
precisely, the value of the region's vector clock at its index is
compared to $t_s$ at the same index.  Snapshot $s$ is consistent iff
$t_s = (C_1[1], ..., C_n[n])$ \citep{virtual}.  Fig.~\ref{fig:vccons}
shows an example for a consistent snapshot. Comparing the global time
to the regions' vector clocks shows that for all memory regions the
value at the respective index is equal to the global time vector at
the same index. This means that for all regions the causing event of
the latest access on them is included in the
snapshot. Fig.~\ref{fig:vcincons} shows an example for an inconsistent
snapshot. In this case the last event on $r_3$ is missing from the
snapshot.  This is a problem as the last event on $r_2$ which is
included in the snapshot is causally dependent on the
event. Therefore, the vector clock has not been updated yet and the
inconsistency can be identified by comparing the vector clock to the
global time vector.

\begin{fullversion}
\subsection{Test program}\label{sec:testprog}

In the test program memory regions are allocated and the causal relationships
between them are tracked using vector clocks. The causal relationships are created
by multiple threads that access the regions and change their contents.
Multiple ways to implement something like this exist, one
possibility is a linked list which is modified by multiple threads. The
elements of the linked list are the memory regions that are being modified. They
are linked with pointers to the previous and next element. Additionally, each list
element has a buffer to store data, and a vector clock to keep track of the
causal dependencies between the list elements. Multiple threads
remove elements from the list, add them in new places and change the contents of
the list elements' buffers. For simplicity the number of elements and threads is
static. When a thread makes a change
to a list element, it keeps a copy of the vector clock. On the next change the
thread makes to a list element this vector clock is used to update the values of
the list element's vector clock.

\subsection{Planned evaluation}\label{sec:eval}
An evaluation using the described test program is planned for the future. The
general workflow of the evaluation
is to execute the test program with varying numbers of list
elements and numbers of threads. After the test program has been running for a
randomized time, a memory snapshot is created.
From the created memoy snapshots the heap of the test program is reconstructed and
the list elements and vector clocks are extracted.
Then inconsistencies can be searched by computing the global time and comparing
each vector clock as described in section \ref{sec:vclocks}. If inconsistencies are
found their impact can be further evaluated by examining the list elements.
Apart from the number of threads and elements in the test program the address
space size and overall system load can be varied.

An overall goal for the evaluation is to create a framework that is easily
reusable by others.
\begin{work}One possibility to ensure this could be to provide Docker containers
that automatically install the necessary software, start the test program and
memory dumping tool. The heap reconstruction can be performed with, for example,
Volatility3. For the extraction of the list elements and vector clocks a python
script has already been written.
\end{work}
\end{fullversion}

\subsection{Using realtime clocks}

While they allow to capture any programmable cause-effect
relationship, vector clocks are rather expensive in terms of memory. A
cheap replacement of vector clocks is to simply take the measurement
of a realtime clock as timestamp (if such a clock exists). With this
idea, the same approach as described above can be used to track the sequence in
which events occurred: Each memory region is
assigned a single realtime
timestamp which corresponds to the time that the most recent event
happened in that memory region. The vector of all such timestamps is
called the \emph{current time}.

A snapshot algorithm now has to keep track of these most recent
timestamps during the acquisition of memory regions. In analogy to the
definitions for vector clocks, the global time $t_s$ of a snapshot $s$
is the vector of these timestamps, one for each memory region.  A
snapshot $s$ is consistent if the current time is equal to the global
time $t_s$. This is a sufficient criterion for quasi-instantaneous consistency
not a necessary one. Finding a \emph{necessary} criterion is an open question.  

Obviously, this method is much more space efficient than using vector
clocks, but it can only be used to check for \emph{quasi-instantaneous}
consistency and not for causal consistency.

\section{Defining Integrity}
\label{sec:integrity}

We briefly revisit the concept of integrity. Integrity wishes to
capture the degree to which a snapshot was influenced by the
measurement method itself. To do this, it is necessary to distinguish
changes on storage that are due to the snapshot mechanism and those
that are not. \citet{correctness} do this by defining a specific point
in time $\tau$ which indicates the ``start'' of the
measurement. Changes before $\tau$ are not due to the measurement
mechanism and changes after $\tau$ affect integrity.

In the definition of \citet{correctness}, a snapshot satisfies
\emph{integrity with respect to $\tau$} if the memory contents
did not change between this point in time and the time of the
acquisition of the region, formally:
\begin{multline*}
  \forall r \in R: \tau \leq s(r).t \implies \forall t' \in T: \\
  \tau \leq t' \leq s(r).t: s(r).v = m(r, t') 
\end{multline*}
With this definition, whenever a memory region's content changes after $\tau$, integrity is not
satisfied anymore. We therefore call it \emph{restrictive integrity}. However,
if the original value is restored before the memory
region is added to the snapshot, then the result is the same as if the change never
happened. 
We therefore propose a slightly weaker definition of integrity, called
\emph{permissive integrity}, that allows changes in memory after $\tau$ as long
as the value that is written to the snapshot is the same as the value that
existed in memory at time $\tau$.

\begin{definition}[permissive integrity]
  A snapshot $s$ satisfies \emph{integrity with respect to time $\tau$} iff
	$$\forall r \in R: \tau\leq s(r).t \implies s(r).v = m(r,\tau) $$
\end{definition}

This definition is more permissive and enables new acquisition
techniques that selectively overwrite memory regions if the
snapshot contains the original data. Obviously, a snapshot the
satisfies restrictive integrity with respect to $\tau$ also satisfies
permissive integrity with respect to $\tau$.


\section{Relations between the Quality Criteria}
\label{sec:relations}

In the original definitions of \citet{correctness}, the three notions
of correctness, atomicity and integrity were not fully independent. In
fact, integrity appeared to be unnecessarily strong and complex: A
snapshot that satisfied integrity also satisfied atomicity and
correctness. From a conceptual point of view, it is better to have
definitions that do not imply each other to separate
concerns.

\begin{figure*}
  \centering
  \includegraphics[scale=0.9]{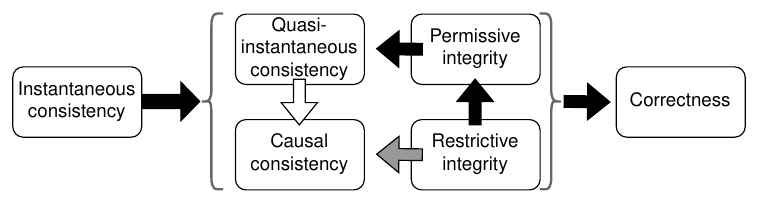}
  \caption{The relationships between the different quality criteria.
    Black arrows indicate implications without further assumptions.
    The gray arrow stands for an implication that only exists if we
    assume that we only observe causal relationships between modifying
    events. The implication shown by the white arrow assume that all
    events are uniquely modifying.}
  \label{fig:bigpicture}
\end{figure*}

Fig.~\ref{fig:bigpicture} shows an overview of the quality criteria
and their relations with respect to implication.  The implications
between the different consistency definitions (see
section~\ref{sec:qualcrit}) and between the two integrity definitions
(see section~\ref{sec:integrity}) are already integrated. As might be
expected, instantaneous consistency and restrictive integrity are the
strongest notions and do not imply each other in any way. Weaker
consistency and integrity definitions are, however, not so easily
separable. As already observed above, their relations also depend on
further assumptions about the observability of events.

We first note that restrictive integrity implies causal consistency
under the assumption that only modifying events are observed. This is
because restrictive integrity, for every memory region, disallows
\emph{any} state changes between $\tau$ and the time the snapshot is
taken. Therefore, if all events are modifying, no event can happen
between $\tau$ and the snapshot, and therefore the snapshot must be
causally consistent. 

The relation between permissive integrity and quasi-instantaneous
consistency is particularly delicate.  Fig.~\ref{fig:quasinotint}
shows a quasi-instantaneous snapshot that does not satisfy permissive
integrity: The snapshot does not satisfy (permissive) integrity with
regard to $\tau$ but a point in time $t$ can be found at which an
instantaneous snapshot would have contained the same values. So
quasi-instantaneous consistency does not imply permissive
integrity. The inverse, however, is true.

\begin{figure}
  \centering
  \includegraphics[scale=0.95]{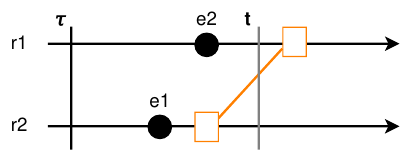}
  \caption{While this snapshot does not satisfy integrity with respect
    to $\tau$ it still satisfies quasi-instantaneous consistency
    because an instantaneous snapshot at time t would have contained
    the same values.}
  \label{fig:quasinotint}
\end{figure}

\begin{proposition}
  \label{prop:2}
  Every snapshot that satisfies permissive integrity with respect to
  $\tau$ also satisfies quasi-instantaneous consistency.
\end{proposition}

\begin{proof}
  Let $s$ be a snapshot that satisfies permissive integrity with
  respect to $\tau$. From the definition this implies that
  $$\forall r\in R: \tau\leq s(r).t \Rightarrow s(r).v=m(r,\tau).$$
  Now consider the instantaneous snapshot $s'$ taken at time
  $\tau$. Since $s'$ was taken at time $\tau$, for all memory regions
  $r$ holds that $s'(r).t=\tau$ and $s'(r).v = m(r,\tau)$. This means
  that $s$ and $s'$ were taken at different times but contain the same
  values, namely the values stored in memory at time $\tau$. So
  overall there exists an instantaneous snapshot that has the same
  values as $s$. Therefore, $s$ is quasi-instantaneous.
\end{proof}

Note that Proposition~\ref{prop:2} holds without any further
assumptions about the observability of events. Therefore, restrictive
integrity also implies quasi-instantaneous consistency without any
further assumptions.  If we assume that we have only uniquely
modifying events, both quasi-instantaneous consistency and permissive
integrity imply causal consistency. Integrity and consistency therefore seem
hard to disentangle completely from each other. If events do not necessarily 
change the values of memory regions, then permissive integrity and causal
consistency are independent of each other.

Both integrity definitions imply correctness because they compare the
contents of the snapshot with the contents of memory. If the
acquisition method were functioning incorrectly this comparison would
be likely to fail. This fact shows that correctness is not really
necessary as an independent concept. Integrity and consistency suffice
to determine the quality of snapshots.

\section{Legal Implications}
\label{sec:legal}



Knowing about the quality of the memory snapshots produced by different tools under certain
circumstances can help investigators to choose the tool best suited for a
concrete investigation. But does it also oblige them to use the best available
tool?

When we try to answer this question, we have to think about the evidential value
of the memory snapshot. Because the quality of the memory snapshot influences the
reliability and completeness of the subsequent analysis results, their
evidential value is also influenced by the memory snapshot's evidential value.
The value a piece of evidence has is equal to the probability that deductions
based on it will be true \citep{heinson2016}.
As the evidence a court grounds its decision in has to be of the highest possible
quality to justify a conviction \citep{kk}[\S 261 recital 5 ff.]
investigators should strive for gathering evidence with an as high as possible
evidential value.
The evidential value of data is determined by the forensic process with which it was
gathered. Among others its authenticity, and integrity as well as the
reliability of used methods should be ensured \citep{safeguarding}.
Tools that are known to produce incorrect memory snapshots must be excluded from
usage in an investigation because this would also cast doubt on any analysis results
and conclusions derived from them, their evidential value would be too low to
justify a conviction.

In the case of integrity and atomicity a
closer look is needed. A tool that produces a snapshot with low integrity overwrites more parts of
the memory than a tool that produces a snapshot with a higher degree of integrity.
Therefore, loss of information in the memory snapshot is more likely if less
integrity can be achieved. As the presented evidence should also be as authentic
as possible \citep{safeguarding} the method that extracts memory snapshots with
higher integrity should be chosen if it also produces correct ones. 

Because
less atomic snapshots are also more likely to have inconsistencies than more atomic
snapshots, the reliability and completeness of the results of an analysis of such a
snapshot can be questioned. Therefore, in trying to adhere to the quality requirements of
evidence used in court decisions, the more atomic method should be chosen if
possible. It should also fulfill the requirements regarding correctness and
integrity.
Another influencing factor on the evidential value of results based on less
atomic memory snapshots would be how likely it is that inconsistencies in memory snapshots lead
to analysis results that suggest the presence of incriminating evidence
even though it never existed in memory. To the best of the authors' knowledge no
research has been published about this topic.


Tools with atomicity guarantees, be they instantaneous,
quasi-instantaneous or causal consistency, can often not be used due to the
technical circumstances of the investigation and time constraints. If a tool
without atomicity guarantees is used, many
inconsistencies might occur. The information how likely their occurrence is for a
specific snapshot is helpful because investigators or expert witnesses who present
the results of a technical
analysis need to explain the likeliness of errors or missing information to the
court. The court should also be enabled to evaluate how likely different
hypotheses based on the presented evidence are and if the evidence is reliable
\citep{safeguarding}.
While it is possible to find some inconsistencies by examining the data structures of the
operating system, thereby enabling analysts to report them exactly, others might
not be visible. Therefore indicators for the
likeliness of the occurrence of inconsistencies, like for example suggested by
\citet{introducing}, should be made available by the memory snapshoting
tool to the analyst. This would enable analysts to provide founded estimates
about the likeliness of the analysis results being incomplete or the possibility
of wrong results due to inconsistencies.

%


\section{Conclusions and Future Work}
\label{sec:concl}

The new notions of atomicity and integrity wish to clarify the
conditions under which snapshots of storage can be considered as
``good''. The definitions assume a synchronous system but cover any
form of storage which cannot be ``frozen'' and where individual memory
regions have to be acquired sequentially.

The measurement approach described in Section~\ref{sec:measuring}
needs to be evaluated in future work. The question remains how the
results for a subset of memory regions can be transferred to the
quality of the complete memory snapshot. Therefore, it will be
necessary to perform an evaluation of the method itself before testing
different memory dumping tools. To evaluate the method, the same steps
as for a tool evaluation can be performed. Memory snapshots are
created while the test program is running and, with the help of vector
clocks, atomicity violations are identified.  Then, other indicators
for inconsistencies, like those described by \citet{introducing}, need
to be examined.  A ground truth of the memory state might be helpful
to identify further inconsistencies. The ground truth could be created
in a virtualized environment by taking atomic snapshots from the
hypervisor.  Identifying as many indicators for inconsistencies as
possible and creating a big data set of analyzed memory snapshots is
another challenge. Indicators can be identified from the related
literature.

The creation of a big data set requires the automation of memory
snapshot creation and the analysis of memory snapshots, and the
organization of the analysis results. A big data set enables
statistical analysis with which it can be evaluated, e.g., if the
number of atomicity violations in a subset of memory can be used to
extrapolate the occurrence of other inconsistency indicators in other
memory areas.  Another question that seems worth examining is if
scenarios can be observed in which inconsistencies lead to false
conclusions that create incriminating evidence where none is present.

\subsection*{Acknowledgments}

We thank Nicole Scheler and Ralph Palutke for helpful comments on
previous versions of this paper. Work was supported by Deutsche
Forschungsgemeinschaft (DFG, German Research Foundation) as part of
the Research and Training Group 2475 ``Cybercrime and Forensic
Computing'' (grant number 393541319/GRK2475/1-2019).

\bibliography{quellen}

\printcredits

\end{document}